\documentclass[11pt]{article}

\usepackage{latexsym}
\usepackage{amsfonts,amssymb}
\usepackage{graphicx, gastex}
\usepackage{color}
\usepackage{multirow}
\usepackage{bigstrut}
\usepackage{array}
\usepackage{makebox}
\usepackage{algorithm}
\usepackage{algpseudocode}
\usepackage{listings}
\usepackage{textcomp}
\usepackage{cases}
\usepackage{setspace}
\usepackage{lipsum}
\usepackage{amsthm}
\usepackage{amssymb}

\usepackage{url}
\urldef\myurl\url{https://dl.dropbox.com/u/50248089/Ozsim.py}
\urldef\myurlcuda\url{https://dl.dropbox.com/u/50248089/Cozsim.py}

\newtheorem{theorem}{Theorem}[section]
\newtheorem{prop}[theorem]{Proposition}

\title{A multi-lane traffic simulation model via continuous cellular automata}

\author{{\bf Emanuele Rodaro}\\ $ $\\ {\em Centro de
Matem\'{a}tica, Faculdade de Ci\^{e}ncias, Universidade do Porto,}\\ 
{\em R. Campo Alegre 687, 4169-007 Porto, Portugal}\\
{\em e-mail:} emanuele.rodaro@fc.up.pt\\$ $\\
{\bf  \"{O}znur Yeldan}\\ $ $\\
{\em Dipartimento di Matematica, Politecnico di Milano,}\\
{\em Piazza Leonardo da Vinci 32, 20133 Milano, Italy}\\
{\em e-mail:} ozyeldan@gmail.com
}

\date{\today}

\begin{document}
\maketitle

\begin{abstract}
Traffic models based on cellular automata have high computational efficiency because of their simplicity in describing unrealistic vehicular behavior and the versatility of cellular automata to be implemented on parallel processing. On the other hand, the other microscopic traffic models such as car-following models are computationally more expensive, but they have more realistic driver behaviors and detailed vehicle characteristics. We propose a new class between these two categories, defining a traffic model based on continuous cellular automata where we combine the efficiency of cellular automata models with the accuracy of the other microscopic models. More precisely, we introduce a stochastic cellular automata traffic model in which the space is not coarse-grain but continuous. The continuity also allows us to embed a multi-agent fuzzy system proposed to handle uncertainties in decision making on road traffic. Therefore, we simulate different driver behaviors and study the effect of various compositions of vehicles within the traffic stream from the macroscopic point of view. The experimental results show that our model is able to reproduce the typical traffic flow phenomena showing a variety of effects due to the heterogeneity of traffic.
\end{abstract}

\section{Introduction}
Traffic models are fundamental resources in the management of road network. A real progress in the study of traffic has obtained with the introduction of the models based on cellular automata (CA)\footnote{Throughout this paper, abbreviation CA refers to both cellular automata (plural) and cellular automaton (singular).}. A cellular automaton is a collection of cells (sites) on a grid of specified shape (lattice) that evolves through a number of discrete time steps according to a set of local rules based on the states of neighboring cells, for more details see \cite{kari}. CA models (CAMs) have the ability of being easily implemented for parallel computing because of their intrinsic synchronous behavior. A set of simple rules can be used to simulate a complex behavior, thus these models are conceptually simple (for the theory and applications of CA as models for complex systems in several scientific fields, see \cite{chaudhuri,chopard,farina-dennunzio}). The traffic models based on CA are capable of capturing micro-level dynamics and relating these to macro-level traffic flow behavior. However, they are not as accurate as the usual microscopic traffic models\footnote{We use the term ``usual microscopic traffic models" to refer to the microscopic traffic models other than CA, which are in general defined by means of a system of differential equations.} such as the time-continuous car-following ones \cite{brackstone-mcdonald}. A basic one-dimensional CAM for highway traffic flow was first introduced by Wolfram, where he gave an extensive classification of CAMs as mathematical models for self-organizing dynamic systems \cite{wolfram,wolfram2}. Nagel and Schreckenberg proposed the first nontrivial traffic model (the NaSch model) based on CA for single-lane highway in 1992 \cite{nasch}. This is a time- and space-discrete model where the traffic road is divided into cells of $7.5$ $m$ and it is defined on a one-dimensional array of a fixed number of sites with closed (periodic) boundary conditions. The velocity is expressed as the number of cells that a vehicle advances in one time step. The maximum velocity ($v_{max}$) is assumed to be $5$ $cells/sec$ ($135$ $km/h$) and every vehicle has the same target velocity $v_{max}$. Each update of the movements of the vehicles is determined by four consecutive rules that are performed in parallel to each vehicle at each second. In literature, the CA traffic models are updated similarly where they only differ on the randomization step introduced to add some randomness, for a survey see \cite{maerivoet-deMoor}. Note that in all of them the cells represent the space as it is in the NaSch model, so in the sequel we call them as ``NaSch-type" models \cite{knospe-santen-schad-schrec,nagel-wolf-wagner-simon,rick-nagel-sch-lat,wagner-nagel-wolf}.

In this paper, we propose a new traffic model via continuous cellular automata (CCA) which gets closer to the car-following models by introducing some continuity without losing the computational advantages typical of CAMs. Continuous cellular automata (or coupled map lattices) are cellular automata where the states of the cells are real values in $[0,1]$, and the local transition rule is a real function. They have have been introduced by Kaneko as simple models with the features of spatiotemporal chaos, and have applications in many different areas like fluid dynamics, biology, chemistry, etc. (for more details on couple map lattices, see \cite{kaneko,keller-kunzle-nowiki}). In this work we drop the restriction that the set of states is $[0,1]$, and we use the name CCA to consider CA where the set of state is in general an infinite set of the same cardinality of the continuum. Therefore, we consider a hybrid between the usual microscopic traffic models which are very accurate in predicting general traffic behavior but computationally expensive, and the CAMs which are very efficient due to their simplicity and intrinsic parallelism making them natural to be implemented for parallel computing. This process of passing from the typical coarse-granularity of CAMs to the continuity in space is done with a change of vision where we abandon the cell-space correspondence and assume that cells represent vehicles. Thus we consider a model with open boundary conditions where the number of cells is equal to the number of vehicles. In this way, we obtain the immediate advantage of having less cells to compute (just the number of vehicles). The continuity also gives us the possibility to refine the microscopic rules that govern the traffic dynamics, using fuzzy reasoning to mimic different real-world driver behaviors. All parameters of the decision process of the drivers are modeled individually by means of fuzzy subsets so that various types of drivers can be taken into consideration. Hence, we are able to study how the heterogeneity influences the traffic macroscopically. The CCA model proposed here is defined first for a single-lane road and then we extend it to the multi-lane case. The assumption that cells represent vehicles makes this extension not as natural as it is in the NaSch-type models.

\section{The single-lane model}\label{sec:ourCA}

In our CCA model, for the purpose of simulating different driver (vehicle) behaviors, we apply a multi-agent fuzzy system where the environmental information is transformed into the decision for the next time step action for each vehicle. Using this fuzzy logic-based system, we categorize the vehicles into types (kinds) where they have common characteristics such as the same perception of time or distance. In this way, we aim to study the effect of different composition of vehicles, i.e., the heterogeneity in road traffic. Let us now describe our model for a single-lane road. Note that for the sake of simplicity, we fix the unit of time as one second. Consider the one dimensional CCA  $\mathcal{SL}=(\mathbb{Z},\Sigma,\mathcal{N},\delta)$ where the lattice is the set of integers and the set of cell states $\Sigma = (K\times\mathbb{R}_{0}^{+}\times\mathbb{R}_{0}^{+}\times\mathbb{R}\times\left\{L,0,R\right\} \times \left\{L,0,R\right\})\cup\{\perp \}$. A cell with the state $\perp$ represents a cell without a vehicle. The generic $i$-th non-empty cell is in the state $\sigma_{i}(t)=(k_{i},x_{i}(t),v_{i}(t),s_{i}(t),d_{i}(t),d_{i}'(t))$ where
\begin{itemize}
    \item $k_{i}$ represents the kind consisting of all the information (parameters) specified differently for each kind of vehicle, such as: the maximum velocity $v_{max_{i}}$, the optimal velocity $v_{opt_{i}}$ (the velocity which each kind of vehicle feels comfortable in traffic stream), the length $l_i$, the fuzzy membership functions, the maximum stress $s_{max_{i}}$, the minimum stress $s_{min_{i}}$, the probability functions of lane-changing to the right lane $P_{R_{i}}(x)$ and to the left lane $P_{L_{i}}(x)$.
    \item $x_{i}(t)$ is the position, defined as the distance from the origin of the road to the middle point of this vehicle, $v_{i}(t)$ is the velocity, and $s_{i}(t)$ is the stress, a variable to keep track of how much the driver is above or below of his optimal velocity. The stress parameter is introduced to implement a more realistic driver behavior since we make the assumption that drivers usually tend to decelerate when they are moving with a velocity higher than their optimal velocity. We assume that if $s_i(t)<0$ then the driver is stressed and wants to go faster, instead if $s_i(t)\ge 0$ then his stress derives from the desire to go slower. Note that the main usage of $s_{i}(t)$ is related to the lane-changing process described in Section \ref{sec:multilane}.
    \item $d_{i}(t)$ is the variable describing the desire for lane-changing to the left ``$L$", to the right ``$R$" and staying on the same lane ``$0$". $d_{i}'(t)$ is the variable showing from which lane the $i$-th vehicle is transferred: from the left lane ``$L$", from the right lane ``$R$"  or not transferred ``$0$". Although these variables are used in the multi-lane model, the evaluations are done in the single-lane part.
\end{itemize}
$\mathcal{N}$ is a kind of one-dimensional extended Moore neighborhood defined by $\mathcal{N}(i)=(i-1,i,i+1,i+2)$, and $\delta :\Sigma^{4}\rightarrow\Sigma$ is the local transition function (local rule) acting upon a cell and its direct neighborhoods. The local rule is defined componentwise by $\delta(\sigma_{i-1}(t),\sigma_{i}(t),\sigma_{i+1}(t),\sigma_{i+2}(t))=(k_{i},x_{i}(t+1),v_{i}(t+1),s_{i}(t+1),d_{i}(t+1),0)$. The variables $s_{i}(t+1)$ and $d_{i}(t+1)$ are described in Section \ref{sec:multilane}. The space is updated as usual by $x_{i}(t+1)=x_{i}(t)+v_{i}(t+1)$ and the velocity by $v_{i}(t+1)=\min (v_{max_{i}},\Delta x_{i}^{+}(t),\max(0,v_{i}(t)+A_{i}(t)))$ where $A_i(t)$ is the acceleration calculated using the fuzzy decision modules and $\Delta x_{i}^{+}(t)$ is the distance with the front vehicle from front bumper to rear bumper defined in (\ref{eq:front distance}). This front distance constraint in the choice of the updated velocity is introduced to make the model collision-free as in the NaSch model. However, since it is unrealistic, it is applied in our model only in the borderline situations where extreme decelerations are involved. Indeed, we have tried to avoid this constraint as much as possible by introducing enough fuzzy rules to make the system more reactive to reduce these extreme situations. It is worth noting that in the case $A_i(t)=7.5$ $m/s^2$, we essentially obtain the open boundary, space-continuous and deterministic version of the NaSch model \cite{nasch}.

$A_{i}(t)$ is depending on the kind $k_{i}$, the velocity $v_{i}(t)$ and the variables defined as following.
\begin{enumerate}
\item Back Distance ($BD$): $\Delta x_{i}^{-}(t)=x_{i}(t)-x_{i-1}(t)-\frac{l_{i}}{2}-\frac{l_{i-1}}{2}$,
\item Front Distance ($FD$):
\begin{equation}\label{eq:front distance}
\Delta x_{i}^{+}(t)=x_{i+1}(t)-x_{i}(t)-\frac{l_{i+1}}{2}-\frac{l_{i}}{2},
\end{equation}
\item Next Front Distance ($NFD$): $\Delta x_{i,N}^{+}(t)=x_{i+2}(t)-x_{i}(t)-\frac{l_{i+2}}{2}-\frac{l_{i}}{2}$. This is introduced to have a better perception of the driver behavior of the next front vehicle since in reality drivers observe not only the vehicle just in front of them but also the vehicles ahead.
\item Perceived Front Collision Time ($PFCT$):
$$
\tau_{i,P}^{+}(t)=\left\{
  \begin{array}{ll}
    \zeta_{i}(t), & \quad \hbox{if}\;\tau_{i}^{+}(t)<0,\\
    \min(\zeta_{i}(t),\tau_{i}^{+}(t)), & \quad \hbox{otherwise.}
  \end{array}
\right.
$$
where $\zeta_{i}(t)$ is the parameter introduced for slowing down depending on the stress calculated by $\zeta_{i}(t)=\frac{s_{max_{i}}-s_{i}(t)}{v_{i}(t)}$ and $\tau_{i}^{+}(t)$ is Front Collision Time ($FCT$), the time that
passes for the $i$-th vehicle to reach to (to collide with) the front vehicle, defined as $\tau_{i}^{+}(t)=\frac{\Delta x_{i}^{+}(t)}{v_{i}(t)-v_{i+1}(t)}$. $PFCT$ is a parameter which is a combination between FCT and an auxiliary time defined to keep the velocity close to the optimal velocity using a parameter depending on the stress variable. If the FCT is strictly negative, then the driver perceives an artificial vehicle with the FCT $\zeta_{i}(t)$ to keep the vehicle close to its optimal velocity. Otherwise, $PFCT$ takes the smallest FCT value between $\zeta_{i}(t)$ and $\tau_{i}^{+}(t)$ to avoid collisions.
\item Worst Front Collision Time ($WFCT$): $\tau _{i,W}^{+}(t)=\frac{\Delta x_{i}^{+}(t)}{v_{i}(t)}$.
The collision time with the front vehicle in the case it suddenly stops. This parameter is introduced for safety reasons.
\item Next Front Collision Time ($NFCT$): $\tau_{i,N}^{+}(t)=\frac{\Delta x_{i,N}^{+}(t)}{v_{i}(t)-v_{i+2}(t)}$. The collision time with the next front vehicle. It is clear that a vehicle never reaches to its next front vehicle, but $NFCT$ is introduced to anticipate the braking maneuver of next front vehicle.
\item Back Collision Time ($BCT$): $\tau_{i}^{-}(t)=\frac{\Delta x_{i}^{-}(t)}{v_{i-1}(t)-v_{i}(t)}$. The collision time of the back vehicle. This is introduced to take into account the phenomenon where the back driver comes closer the front driver and forces him to accelerate. We call this situation as \textit{pushing effect}.
\end{enumerate}

The approach of embedding fuzzy logic while dealing with a system described by continuous variables has been already introduced, indeed, there are several works based on fuzzy logic systems in car-following models \cite{brackstone-mcdonald}. In our model, we use a fuzzy rule-based system by means of fuzzy IF-THEN rules \cite{mamdani} that determines the behavior of the vehicles in traffic stream. In this fuzzy system we introduce two fuzzy modules, formed by two sets of fuzzy rules, according to their significance. The first module has more importance since it is based on the information related to the front and the back vehicles while the second one has less importance since the rules in that module are related to the next front vehicle. As a consequence of having two different fuzzy modules, we have two outputs: $A_{i,1}$ is the output coming out from the first fuzzy module and $A_{i,2}$ is the output coming out from the second fuzzy module. The first and the seconds sets of fuzzy rules used in the simulator are summarized in Table \ref{tab:firstmodule1part} and Table \ref{tab:secondmodule}, respectively. After receiving the inputs $\tau_{i,P}^{+}(t)$, $\tau _{i,W}^{+}(t)$, $\tau _{i,N}^{+}(t)$, $\tau _{i}^{-}(t)$, $\Delta x_{i}^{+}(t)$, $\Delta x_{i,N}^{+}(t)$, $\Delta x_{i}^{-}(t)$, $v_{i}(t)$ determined by the environment, the final decision of the acceleration $A_{i}(t)$ is returned by the function (to avoid cumbersome notation, we omit the dependency of $t$):
\begin{equation}\label{func: evaluation of final decision}
F(A_{i,1},A_{i,2})=\left\{
\begin{array}{ll}
\min (A_{i,1},A_{i,2}) & \hbox{if}\;A_{i,1}\leq 0, \\
\frac{A_{i,1}+A_{i,2}}{2} & \hbox{if}\;A_{i,1}>0\wedge A_{i,2}\leq -0.25, \\
A_{i,1} & \hbox{otherwise.}
\end{array}
\right.
\end{equation}
where we give more weight to the decision taken by the first module. Indeed, in the simulation, we have noticed that without these kind of constraints the vehicles were tending to slow down too much. This is also the reason of splitting the fuzzy module into two parts. Moreover, the second module and consequently the second acceleration gains significance only in the case of emergencies such as a sudden breakdown or deceleration of the next front vehicle. The linguistic terms (properties) used for the acceleration outputs $A_i$ are: PB=Positive Big, PM=Positive Medium, PS=Positive Small, Z=Zero, NS=Negative Small, NM=Negative Medium, and NB=Negative Big.
\begin{table}[h]
\caption{First fuzzy set of rules}
\begin{center}
\scalebox{0.6}{
\begin{tabular}{ccc|c|c|c||c|l}
\cline{3-7}
\multicolumn{1}{c}{} & &
\multicolumn{4}{|c||}{Front Distance (FD)}
& \multicolumn{1}{c|}{Velocity ($v_{i}$)} \bigstrut \\ \cline{3-7}
\multicolumn{1}{c}{} & &
\multicolumn{1}{|c|}{BIG} & MEDIUM & SMALL & VERY SMALL & SMALL & \bigstrut \\ \cline{1-7}
\multicolumn{1}{|c|}{\multirow{4}{*}{PFCT}} &
\multicolumn{1}{|c|}{\multirow{2}{*}{BIG}} &
\multicolumn{1}{c|}{($v_{i}$ NOT SMALL)} & $(v_{i}$ NOT SMALL) & ZERO & ZERO & PB &  \bigstrut \\
\multicolumn{1}{|c|}{} & &
\multicolumn{1}{|c|}{PM} & PS &  &  & (Jam Situation) \bigstrut \\ \cline{2-7}
\multicolumn{1}{|c|}{} &
\multicolumn{1}{|c|}{MEDIUM} & ZERO & ZERO & NS & NS &
\multicolumn{1}{c}{} \bigstrut \\ \cline{2-6}
\multicolumn{1}{|c|}{} &
\multicolumn{1}{|c|}{SMALL} & NM & NM & NM & NM &
\multicolumn{1}{c}{} \bigstrut \\ \cline{2-6}
\multicolumn{1}{|c|}{} &
\multicolumn{1}{|c|}{VERY SMALL} & NB & NB & NB & NB &
\multicolumn{1}{c}{} \bigstrut \\ \cline{1-6}
\\ \cline{1-6}
\multicolumn{1}{|c|}{\multirow{1}{*}{WCT}} &
\multicolumn{1}{|c|}{SMALL} &
\multicolumn{1}{|c|}{} &
\multicolumn{1}{c|}{NS} & NM & NM &
\multicolumn{1}{c}{}  \bigstrut \\ \cline{1-6}
\end{tabular}
}
\end{center}
\label{tab:firstmodule1part}
\end{table}

In the first fuzzy set of rules we have also the ones related to the back vehicle, thus we consider the pushing effect where $BCT$ is ``Very Small" and $BD$ is ``Very Small". In this case, we apply the following rule formed by taking into consideration the safety with the front vehicle:
(IF $PFCT$ is $BIG$ AND $FD$ is $BIG$) OR (IF $PFCT$ is $BIG$ AND $FD$ is $MEDIUM$) OR (IF $PFCT$ is $MEDIUM$ AND $FD$ is $BIG$) OR (IF $PFCT$ is $MEDIUM$ AND $FD$ is $MEDIUM$) then the acceleration choice is $PS.$
\begin{table}[h]
\caption{Second fuzzy set of rules}
\begin{center}
\scalebox{0.6}{
    \begin{tabular}{cc | c | c | c | c |}
    \cline{3-6}
     & & \multicolumn{4}{c|}{Next Front Distance (NFD)} \bigstrut \\ \cline{3-6}
     & & VERY SMALL & SMALL & MEDIUM & BIG \bigstrut \\ \cline{1-6}
    \multicolumn{1}{|c|}{\multirow{4}{*}{NFCT}} &
    \multicolumn{1}{c|}{VERY SMALL} & NB & NB & NM & NM  \bigstrut  \\ \cline{2-6}
    \multicolumn{1}{|c|}{} &
    \multicolumn{1}{c|}{SMALL} & NM & NM & NS & NS  \bigstrut \\ \cline{2-6}
    \multicolumn{1}{|c|}{} &
    \multicolumn{1}{c|}{MEDIUM} & NS &
    \multicolumn{1}{c}{} &
    \multicolumn{1}{c}{} &
    \multicolumn{1}{c}{}    \bigstrut \\ \cline{2-3}
    \multicolumn{1}{|c|}{} &
    \multicolumn{1}{c|}{BIG} & NS &
    \multicolumn{1}{c}{} &
    \multicolumn{1}{c}{} &
    \multicolumn{1}{c}{}    \bigstrut \\ \cline{1-3}
   \end{tabular}
}
\end{center}
\label{tab:secondmodule}
\end{table}

These rules are based on some common sense of driver behaviors including some experiences. Although the rules are the same for each kind of vehicle, they have different ``weights" depending on the definition of the membership functions of different kinds. In this way, it is possible to give a description of a variety of behaviors, such as the behavior of a long vehicle driver or a driver with low reflexes. For instance, a person with low reflexes perceives a time of collision of $5$ seconds as a very short time, however a person with higher reflexes probably would feel comfortable with that time of collision.


After the fuzzy system receives all the inputs, we first determine the degree of membership of each input. Then by taking the minimum value of the images of the inputs (corresponding to the fuzzy logic AND operation) in each rule, we obtain the weights of the fuzzy rules. This evaluation is described in general terms as following.

Let $B_{1}^{j},...,B_{k_{j}}^{j}$ and $C^{j}$ be fuzzy subsets with the membership functions $\mu_{B_{1}^{j}},\ldots,\mu_{B_{k_{j}}^{j}}$ and $\mu_{C^{j}}$, respectively, and let $R^{j}$ be the fuzzy rules defined as:
$$
R^{j}:IF\;x_{1}^{j}\;\hbox{is}\;B_{1}^{j}\;AND\;\ldots\;AND\;x_{k_{j}}^{j}\;\hbox{is}\;B_{k_{j}}^{j}\;THEN \;y\;is\;C^{j}
$$
with
$$
\mu_{B_{1}^{j}\;}(x_{1}^{j})\wedge \;\ldots\; \wedge \mu_{B_{k_{j}}^{j}\;}(x_{k_{j}}^{j})
=\min \left\{\mu_{B_{1}^{j}\;}(x_{1}^{j}),\;\ldots\;,\mu_{B_{k_{j}}^{j}\;}(x_{k_{j}}^{j})\right\}
$$
for $1\leq j\leq m$, where $m$ is the number of fuzzy rules.

In our fuzzy system, there are also rules including the form: $x_r^{j}\hbox{ is not }B_r^{j}$ with the degree of membership $\mu_{\sim B_r^{j}\;}(x_r^{j})$ for some $r\in [1,k_j]$. These rules which are the first two rules of the first fuzzy decision module, are formed with the purpose of emphasizing the vehicle is not in a jam situation (see Table \ref{tab:firstmodule1part} where $v_i$ is not small). For an input $x_{r}^{j}$, the degree of membership of not having a property $B_r^{j}$ (``not being $B_r^{j}$") can be written as:
$$
\mu_{\sim B_r^{j}}(x_r^{j})=1-\mu_{B_r^{j}}(x_r^{j}).
$$

The output of the previous process so far is a fuzzy set, so we should convert our fuzzy output set into one single number as the output of the fuzzy system, which is the ``acceleration" in our case. Recall that since we have two fuzzy modules, there are two outputs of the system: $A_{i,1}$ and $A_{i,2}$ which are used to evaluate $A_{i}(t)$ at time $t$ by the function $F$, see Equation \ref{func: evaluation of final decision}. The conversion process is called as defuzzification and there are many defuzzification techniques to obtain a crisp value, such as the center-of-gravity (COG) and the weighted average formula (WAF) method. In our model, since we would like to use a simple method that does not require too much computational power, we define a new method and call it as generalized weighted average formula (GWAF), i.e., we generalize the WAF to the case where the membership functions are not necessarily symmetric. We describe it now in general terms. Suppose that each $j$-th rule receives the values $\overline{x}_{1}^{j},...,\overline{x}_{k_{j}}^{j}$ as inputs. Let
$$
w^{j}=\min\left\{ \mu _{B_{1}^{j}\;}(\overline{x}_{1}^{j}),\;...\;,\mu _{B_{k_{j}}^{j}\;}(\overline{x}_{k_{j}}^{j})\right\}
$$
be the weights for each rule $j$ and let $P^{j}=\mu _{C^{j}}^{-1}(w^{j})$ be the preimage of the weight of $j$-th rule. We assume that $\mu _{C^{j}}$ does not have any plateau, since we do not want $P^{j}$ to contain intervals (this is done to have a discrete set). The defuzzified output is thus calculated by:
$$
\overline{y}=\frac{\sum_{j=1}^{m}w^{j}\sum_{z\in
P^{j}}z}{\sum_{j=1}^{m}\left\vert
P^{j}\right\vert w^{j}}.
$$

\section{The multi-lane model}\label{sec:multilane}

In this section, we extend the single-lane model to the multi-lane case. This extension is not trivial as it is in the NaSch-type models where adding a lane simply means adding an array of cells and where the local transition function can naturally be extended. This is a consequence of having a clear physical interpretation of the model given by the fact that space is represented by cells. However, as a consequence of not having the cell-space correspondence, first of all we do not have the natural order which makes the extension of the local transition function so easy to achieve, and secondly not all the configurations $c:\mathbb{Z}\rightarrow \Sigma$ of our model represent a physical situation. For this reason, we consider just the configurations that have physical meaning, i.e., two adjacent cells $n,n+1$ are in the states where the positions fulfill the constraint $x_n(t)<x_{n+1}(t)$. We denote the set consisting of such configurations by $\mathcal{C}onf_{p}(\mathcal{SL})$. The \emph{support} of a configuration $c\in \mathcal{C}onf_{p}(\mathcal{SL})$ is a maximal interval $[i,j]$ of integers with $i\le j$ such that for all $k<i$ and $k>j$, $c(k)=\perp$.

The union of two arrays of cells representing a road with two lanes is a natural candidate for a multi-lane model. Therefore, we first present our multi-lane model as a union of interacting single-lane CCA where the interaction is given by a transfer operation, and then we prove in Prop. \ref{prop} that this model can actually be simulated by a CCA.

The process of lane-changing behavior consists of several steps. First the driver feels the desire of lane-changing, then if the conditions of the place in the lane where the driver wants to move respect some safety constraints, the driver can finally perform the maneuver. Regarding the decision of performing a lane-changing, we first check if the $i$-th vehicle with the state $\sigma_i(t)=(k_i,x_i(t),v_i(t),s_i(t),d_i(t),d_i'(t))$ desires to change lanes, i.e., $d_i(t)\neq0$. The desire of lane-changing is depending on the stress variable $s_i(t)$. Having positive stress means that the driver is above of his optimal velocity, thus he tends to change lane to the right (slower lane). Otherwise, the driver is nervous and he desires to go faster, thus he tends to change lane to the left (faster lane). The stress cannot be increased (or decreased) arbitrarily, so we define the two parameters $s_{min_{i}}$ and $s_{max_{i}}$ which are the maximum amounts of negative and positive stress that a driver can tolerate. Thus if the updated stress $s_i(t+1)$ exceeds the boundaries $s_{min_{i}}$ and $s_{max_{i}}$ we simply assign it to these limit values. The stress is updated by
\begin{numcases}{s_i(t+1)=}
s_{acc}(t)/2 & if $\tau_{i}^{+}(t)<0$ and $\frac{s_{min_{i}}}{2}<s_{acc}(t)<0$,\label{eq:stress1} \\
s_{acc}(t)\cdot(1+\Phi) & if $\tau _{i}^{+}(t)\ge 0$ and $\frac{s_{min_{i}}}{2}<s_{acc}(t)<0$,\label{eq:stress2} \\
\nonumber s_{acc}(t) & otherwise.
\end{numcases}
\noindent where the accumulated stress is defined by $s_{acc}(t)=s_{i}(t)+(v_{i}(t+1)-v_{opt_{i}})\cdot\mathbb{X}(t)$ and $\mathbb{X}(t)$ is a random variable distributed uniformly, $\mathbb{X}(t)\sim\mathcal{U}(0,1)$. Condition (\ref{eq:stress1}) is used to avoid frequent lane-changing in the case of a jam situation, if the queue is moving. We also embed the strategy of trying to change lanes instead of braking in the case the front vehicle is close and tends to brake. We model this effect with the condition (\ref{eq:stress2}) considering a factor $\Phi$ representing how much the $i$-th vehicle is in the situation where the front vehicle is close and going much more slower than the $i$-th vehicle. $\Phi$ is calculated by the membership value of $(\tau _{i}^{+}(t) \mbox{ is } VERY\;SMALL \wedge \Delta x_{i}^{+}(t) \mbox{ is } MEDIUM) \vee (\tau _{i}^{+}(t) \mbox{ is } VERY\;SMALL \wedge \Delta x_{i}^{+}(t) \mbox{ is } SMALL) \vee (\tau _{i}^{+}(t) \mbox{ is } SMALL \wedge \Delta x_{i}^{+}(t) \mbox{ is } MEDIUM) \vee (\tau _{i}^{+}(t) \mbox{ is } SMALL \wedge \Delta x_{i}^{+}(t) \mbox{ is }\\ SMALL)$. It is used to increase the amount of stress and consequently to have more probability of changing lanes.
\begin{algorithm}[h]
\caption{The pseudo-code for evaluating $Eval_{(\mathcal{L},\mathcal{R})}$.}
\label{alg:Eval}
{\fontsize{8.7}{8.7}\selectfont
\begin{algorithmic}[1]
\Procedure{$Eval_{(\mathcal{L},\mathcal{R})}$}{$k_{i},v_{i}(t),s_{i}(t)$}
\If {$s_{i}(t)\ge 0$}
    \State $ns_i(t) = s_i(t)/s_{max_{i}}$
    \State \textbf{Execute the Bernoulli Trial} $\mathbb{X}\sim \mathcal{B}(2,P_{R_{i}}(ns_i(t))$
        \If {$\mathbb{X} = 1$}
            \State $d_i(t+1)= R$
        \Else
            \State $d_i(t+1)= 0$
        \EndIf
\Else
    \State $ns_i(t) = s_i(t)/s_{min_{i}}$
    \State \textbf{Execute the Bernoulli Trial} $\mathbb{Y}\sim \mathcal{B}(2,P_{L_{i}}(ns_i(t))$
    \If {$\mathbb{Y} = 1$}                   
        \State \textbf{Execute the Bernoulli Trial} $\mathbb{Z}\sim \mathcal{B}(2,\mu_{vels}(v_i(t)))$  
            \If {$\mathbb{Z} = 1$}          
                \If {$\mathcal{R} = 0$}     
                    \State $d_i(t+1)= L $   
                \EndIf
                \If {$\mathcal{L} = 0$}     
                    \State $d_i(t+1)= R $   
                \EndIf
                    \State \textbf{Execute the Bernoulli Trial} $\mathbb{W}\sim \mathcal{B}(2,0.7)$  
                    \If {$\mathbb{W} = 1$}  
                        \State $d_i(t+1)= L$ 
                    \Else
                        \State $d_i(t+1)= R$ 
                    \EndIf
                \Else
                    \State $d_i(t+1)= L$
                \EndIf
    \Else
        \State $d_i(t+1)=0$
    \EndIf
\EndIf
\EndProcedure
\end{algorithmic}
}
\end{algorithm}

We now describe the stochastic process that evaluates the desire of lane-changing for the next time step $d_{i}(t+1)=Eval_{(\mathcal{L},\mathcal{R})}(k_{i},v_{i}(t),s_{i}(t))$ described in Algorithm \ref{alg:Eval}. The decision of this action is made by means of a Bernoulli process $\mathcal{B}(2,p)$. The probabilities of changing lane to the left and to the right, which are in general different, are calculated by two functions $P_{L_{i}}(x): [0,1]\rightarrow [0,1]$ and $P_{R_{i}}(x): [0,1]\rightarrow [0,1]$ contained in the kind $k_{i}$. The reason is that some kinds of vehicles tend to move more to the left lane while some others tend to move more to the right. For instance, long vehicles prefer to go to the right lane more than left, i.e., they tend to stay more on the right lane. The variable used to calculate such probabilities is the normalized stress $ns_i(t)=\max(s_{i}(t)/s_{min_{i}},s_{i}(t)/s_{max_{i}})$. Thus if $s_{i}(t)\ge 0$ then the choice to stay or go to the right is done by $\mathcal{B}(2,P_{R_{i}}(ns_i(t)))$, otherwise by $\mathcal{B}(2,P_{L_{i}}(ns_i(t)))$ to go to the left. If the latter process is resulted as going to the left, we take into consideration also the jam situation where the drivers randomly move to the left or right to find an emptier lane. For this purpose we use the velocity parameter and we need the fuzzy set of ``Velocity SMALL" with the membership function denoted by $\mu_{vels}$ to perform this evaluation which is done by the Bernoulli process $\mathcal{B}(2,\mu_{vels}(v_i(t)))$. If the result of the outcome is positive, meaning that the traffic is indeed jammed, the decision for lane-changing depends on the relative position of the lane, i.e., if the vehicle is in the most left (right) lane then it can only move to the right (left). For this reason, the single-lane model $\mathcal{SL}_{(\mathcal{L},\mathcal{R})}$ depends on two parameters $\mathcal{L},\mathcal{R}$ which are Boolean variables used to describe whether or not there exists a lane on the left or on the right, respectively. Thus if there is no lane on the left (right) of the driver, i.e., $\mathcal{L}=0$ ($\mathcal{R}=0$), he can only move to the right (left) lane. If there are lanes both on the left and right, the choice is obtained using the Bernoulli process $\mathcal{B}(2,0.7)$ where the decision is ``moving to the left lane" with the probability of $0.7$ and ``moving to the right lane" otherwise. These probabilities are decided to be different since we make the assumption that the drivers desiring to go faster usually tend to move to the left lane more than the right even in a situation of traffic congestion.

After the desire of lane-changing, we perform the transfer of a vehicle. The transfer clearly changes the configuration of the single-lane CCA, thus we need to introduce some operations to describe this process. Given a vehicle represented by a state $\sigma=(k,x,v,s,d,d') \in \Sigma\setminus \{\perp\}$ and a configuration $c\in \mathcal{C}onf_{p}(\mathcal{SL})$, the \textit{inserting operator} at position $n\in \mathbb{Z}$ is the function $Ins_{n}:\Sigma\setminus \{\perp\}\times \mathcal{C}onf_{p}(\mathcal{SL})\rightarrow \mathcal{C}onf_{p}(\mathcal{SL})$ that changes the configuration from $c$ into $Ins_{n}(\sigma,c)$ by shifting the states of all cells one step starting from the $n$-th position and setting the state of the $n$-th cell to the value $\sigma$. The right-inverse of $Ins_{n}$ is the \textit{deleting operator} at position $n$ denoted by $Del_{n}:\mathcal{C}onf_{p}(\mathcal{SL})\rightarrow \mathcal{C}onf_{p}(\mathcal{SL})$. The index $n$ gives the position where we insert or delete the content of a cell. The function $Ins_{n}(\sigma,c)$ is well defined only if the insertion of the vehicle still generates a physical configuration. However, before the insertion, we check what is the relative position of the vehicle that wants to perform a lane-changing with respect to the lane that it is going to enter. For this reason, we need to define an \textit{index operator}  $Indx:\Sigma \setminus \{\perp\} \times \mathcal{C}onf_{p}(\mathcal{SL})\rightarrow \mathbb{Z}$ such that $Indx(\sigma,c)$ represents the index where the vehicle with state $\sigma$ would go if we would try to insert it into the configuration $c$.

In our model, the process of lane-changing is based on some safety criterions which check the possibility of executing a lane-changing by considering the situation in the target lane. These criterions, obtained from some simulation experiments, guarantee that after the lane-changing there will be no danger (avoidance of collision) with the back and the following vehicle on the target lane. Consider the vehicle represented by $\sigma$ that wants to enter in a lane with configuration $c$, thus the corresponding position in the target lane is between the $j-1$-th and $j$-th vehicle with velocities $v_{j-1}$ and $v_{j}$ where $j=Indx(\sigma,c)$. Suppose that the back and front distances of this entering vehicle with respect to the $j-1$-th and $j$-th vehicle are denoted by $\Delta^{-}$ and $\Delta^{+}$, respectively, then the safety constraints are:
\begin{equation}
\Delta^{-}>\max(0,v_{j-1}^{1.2}-v+\vert v_{j-1}-v\vert +3),\;\Delta^{+} >\max(0,v^{1.25}-v_{j}+3)  \label{eq1:constraint on velocity}
\end{equation}
The condition for performing a lane-changing to the left (right) lane can be summarized by the Boolean variable $trans_L(\sigma,c)$ ($trans_R(\sigma,c)$) which consists of the safety criterion (\ref{eq1:constraint on velocity}), the desire to change lane to the left (right), i.e., $d = L$ ($d = R$) and not being already transferred, i.e., $d'=0$ (to avoid multiple transfers in one step of simulation). In the process of lane-changing, we also need to keep trace where the vehicle comes from because the transfer of a vehicle is seen as a process consisting of copying and erasing steps. For this purpose we define the copy of $\sigma$ as:
$$
\sigma^{cp}=\left\{
\begin{array}{lr}
(k,x,v,s/5,L,R) & \;\hbox{if }\sigma=(k,x,v,s,L,0), \\
(k,x,v,s/5,R,L) & \;\hbox{if }\sigma=(k,x,v,s,R,0).
\end{array}
\right.
$$
where we decrease the stress from $s$ to $s/5$, firstly because we make the assumption that when a vehicle changes lanes there is a sense of satisfaction reducing the stress, and secondly this change helps to reduce the ping-pong phenomenon \cite{knospe-santen-schad-schrec,rick-nagel-sch-lat} of vehicles continuing to move from one lane to another in case of traffic congestion. If a vehicle with state $\sigma$ desires to go to the left (right) lane, then we copy $\sigma$ into $c$, i.e., we update $c$ into a new configuration denoted by $\sigma \rightarrowtail_L c$ ($\sigma \rightarrowtail_R c$)$\in \mathcal{C}onf_{p}(\mathcal{SL})$ defined by the following equation ($X\in\{L,R\}$):
$$
\sigma \rightarrowtail_X c=\left\{
\begin{array}{ll}
Ins_{j}(\sigma^{cp},c)  & \;\hbox{if }trans_X(\sigma ,c), \mbox{ where }j=Indx(\sigma,c)\\
c & \;\hbox{otherwise.}
\end{array}
\right.
$$
Besides, we need an operator to update a configuration $c'$ in which we have to erase the state that has already been copied into $c$. The erasing procedure is accomplished by changing $c'$ into the new configuration $c'\setminus \omega \in \mathcal{C}onf_{p}(\mathcal{SL})$ defined as:
$$
c'\setminus \omega =\left\{
\begin{array}{ll}
Del_{i}(c') &
\begin{array}{r}
\;\;\hbox{if }(c'(i)=(h,y,w,5r,R,0)\wedge \omega=(h,y,w,r,R,L)) \\
\vee(c'(i)=(h,y,w,5r,L,0)\wedge \omega=(h,y,w,r,L,R)),\\
\mbox{where }i=Indx(\omega,c')
\end{array}
\\
c' & \;\;\;\hbox{otherwise.}
\end{array}
\right.
$$

\begin{algorithm}[h]
\caption{The pseudo-code for one time step evolution of the multi-lane model}
\label{alg:evolve multi}
{\fontsize{8.7}{8.7}\selectfont
\begin{algorithmic}[1]
\Procedure{Update}{$c_0,\ldots, c_{M-1}$}
\For{$i = 0 \to M-1$}
    \If{$i=0$}
    \State $c_1 := (c_0\rightarrowtail_R c_1)$
    \State $c_0 := c_0\setminus (c_0\rightarrowtail_R c_1)$
    \EndIf
    \If{$0<i<M-1$}
    \State $c_{i-1} := (c_i\rightarrowtail_L c_{i-1})$
    \State $c_i := c_i\setminus (c_i\rightarrowtail_L c_{i-1})$
    \State $c_{i+1} := (c_i\rightarrowtail_R c_{i+1})$
    \State $c_i := c_i\setminus (c_i\rightarrowtail_R c_{i+1})$
    \If{$i=1$}
    \State $c_{0} :=  \delta_{(0,1)}^*(c_0)$
    \Else
    \State $c_{i-1} := \delta_{(1,1)}^*(c_{i-1})$
    \EndIf
    \EndIf
    \If{$i=M-1$}
    \State $c_{M-2} := (c_{M-1}\rightarrowtail_L c_{M-2})$
    \State $c_{M-1} := c_{M-1}\setminus (c_{M-1}\rightarrowtail_L c_{M-2})$
    \State $c_{M-2} :=  \delta_{(1,1)}^*(c_{M-2})$
    \State $c_{M-1} :=  \delta_{(1,0)}^*(c_{M-1})$
    \EndIf
\EndFor
\EndProcedure
\end{algorithmic}
}
\end{algorithm}
We extend these functions to the binary operations $\rightarrowtail_L,\rightarrowtail_R$ and $\setminus$ on $\mathcal{C}onf_{p}(\mathcal{SL})$ to describe the transfer process of all the vehicles from one lane to another. We call $\rightarrowtail_L$ and $\rightarrowtail_R$ as the \emph{left} and \emph{right copying} operation, respectively, and $\setminus$ as the \emph{erasing} operation. The extension is made in the following way. Let $c,c'\in \mathcal{C}onf_{p}(\mathcal{SL})$ be two configurations with $c(i)=\omega_{i}\neq \perp$ for $i=m,\ldots, M$ and $c'(i)=\sigma_{i}\neq \perp$ for $i=n,\ldots, N$, where $[m,M],[n,N]$ are the supports of $c,c'$, respectively. We define the configurations $c'\rightarrowtail_L c$ inductively as follows: let $e_{0}=c$ and $e_{k}:=\sigma_{k}\rightarrowtail_L e_{k-1}$ for $k = 1,\ldots, N-n+1$ then we define $c'\rightarrowtail_L c:=e_{N-n+1}$ (analogously for $\rightarrowtail_R$). On the other hand, for $c'\setminus c$, let $g_{0}=c'$ and $g_{k}:=g_{k-1}\setminus \omega _{j_{k}}$ for $k = 1,\ldots, M-m+1$ then $c'\setminus c:=g_{M-m+1}.$ Using these operators it is easy to see that the process of the transfer of vehicles from $c'$ to $c$, where $c$ is the configuration of a lane on the left of the lane with configuration $c'$, can be easily described by transforming $c$ into $c'\rightarrowtail_L c$ and $c'$ into $c'\setminus (c'\rightarrowtail_L c)$.

We now present our multi-lane model using these copying and erasing operations. Suppose that there are $M\ge 2$ lanes on a highway. We model each lane using the single-lane CCA model $\mathcal{SL}$. Recall that $\mathcal{SL}$ depends on the parameters $\mathcal{L},\mathcal{R}$, thus we rewrite it as $\mathcal{SL}_{(\mathcal{L},\mathcal{R})}=(\mathbb{Z},\Sigma,\mathcal{N},\delta_{(\mathcal{L},\mathcal{R})})$. In this way, we can associate to the road the ordered $M$-tuple: 
$$
(\mathcal{SL}_{(0,1)},\mathcal{SL}^1_{(1,1)},\ldots,\mathcal{SL}^{M-2}_{(1,1)},\mathcal{SL}_{(1,0)})
$$
where if $M\ge 3$ we have $M-2$ copies of $\mathcal{SL}_{(1,1)}$, and where $\mathcal{SL}_{(0,1)}$ and $\mathcal{SL}_{(1,0)}$ represent the left- and the right-most lane, respectively. In the case $M=2$, we consider just the pair $(\mathcal{SL}_{(0,1)},\mathcal{SL}_{(1,0)})$. Suppose that these $M$ number of CCA are in the configurations $(c_0,\ldots, c_{M-1})\in\mathcal{C}onf_p(\mathcal{SL})^M$. In our multi-lane model, we scan each lane and we transfer the vehicles to the adjacent lanes. After this process, for each lane we apply the single-lane CCA model to update the configuration and this update is done by means of the global transition function of $\mathcal{SL}_{(a,b)}$ denoted by $\delta_{(a,b)}^*,\;a,b\in \{0,1\}$. In this way, we have a new array of configurations $(c_0',\ldots, c_{M-1}')$, and this process represents $1$ $sec.$ of simulation. Moreover, the order with which the transfer is performed is from the left-most lane to the right-most one, and this is done to satisfy the precedence requirement in European roads. In Algorithm \ref{alg:evolve multi}, it is described the updating process $Update: (c_0,\ldots, c_{M-1})\mapsto(c_0',\ldots, c_{M-1}')$. We now show that Algorithm \ref{alg:evolve multi} can be simulated by a CCA which implies that our multi-lane model is framed as a CCA model.
\begin{prop}\label{prop}
    There is a stochastic CCA multi-lane model which simulates Algorithm \ref{alg:evolve multi}.
\end{prop}
\begin{proof}
  We define the stochastic CCA:
    $$
    \mathcal{ML}=(\mathbb{Z}, \Omega, \mathcal{M}, \Delta)
    $$
    where
    \begin{itemize}
        \item $\Omega = (\mathcal{C}onf_p(\mathcal{SL})\times \{copy, erase\}\times \mathbb{N}^3)\cup\{\perp\}$, where $\perp$ is the state associated to the empty cells representing cells with no lane.
        \item $\mathcal{M}(i)=(i-1,i,i+1)$ is the von Neumann neighborhood.
        \item $\Delta:\Omega^3\rightarrow\Omega$ is the local transition function such that
            $$
                \Delta(\omega_{-1},\omega_0,\omega_1)=\omega_0'
            $$
        defined in Algorithm \ref{alg:Delta eval} with inputs
            $$
                \omega_{j}=(c_j,X_j,M_j,P_j,K_j),\;\;j=-1,0,1.
            $$
    \end{itemize}
If we consider $M$ lanes with the configurations $c_0,\ldots,c_{M-1}$, we associate to $\mathcal{ML}=(\mathbb{Z}, \Omega, \mathcal{M}, \Delta)$ the configuration
$$
\mathcal{C}=(\omega_0,\ldots,\omega_{M-1})
$$
where $\omega_i=(c_i,copy,M,i,0)$ for $i=0,\ldots,M-1$. It is easy to see that applying $2M$ times the global transition function $\Delta^*$ to $\mathcal{C}$, we obtain a new configuration
$$
\Delta^{*^{2M}}(\mathcal{C})=(\omega_0',\ldots,\omega_{M-1}')
$$
with $\omega_i'=(c'_i,copy,M,i,0)$ for $i=0,\ldots,M-1$ where
$$
(c'_0,\ldots,c'_{M-1})=Update(c_0,\ldots,c_{M-1}),
$$
and $Update(c_0,\ldots,c_{M-1})$ is the function defined in Algorithm \ref{alg:evolve multi}.
\end{proof}
        \begin{algorithm}[H]
\caption{The pseudo-code to compute the local transition function $\Delta$.}
\label{alg:Delta eval}
{\fontsize{9}{9}\selectfont
\begin{algorithmic}[1]
\Procedure{$\Delta\;$}{$\omega_{-1}$, $\omega_0$, $\omega_1$}
    \If{$\omega_{0}=\perp$}
    \State $\omega_0'=\perp$
    \Else
       \If{$K_0 = 0$}
           \If{$X_0 = copy$}
               \If{$P_0 = K_0 + 1$}
                    \State {$c_{0}:= (c_{-1}\rightarrowtail_R c_0)$}
               \EndIf
               \State {$X_0 := erase$}
            \Else
               \If{$P_0 = K_0$}
                    \State {$c_{0}:= (c_{0}\setminus c_1)$}
               \EndIf
                    \State {$K_0:= K_0 + 1\; mod\; M_0$}
                    \State {$X_0 := copy$}
                    \State {\textbf{exit}}
            \EndIf
       \EndIf
        \If{$0<K_0<M_0 - 1$}
            \If{$X_0 = copy$}
                \If{$P_0 = K_0 - 1$}
                    \State {$c_{0}:= (c_1\rightarrowtail_L c_0)$}
                \EndIf
                \If{$P_0 = K_0 + 1$}
                    \State {$c_{0}:= (c_{-1}\rightarrowtail_R c_0)$}
                \EndIf
                \State {$X_0 := erase$}
            \Else
                \If{$P_0 = K_0$}
                    \State {$c_{0}:= (c_{0}\setminus c_{-1})$}
                    \State {$c_{0}:= (c_{0}\setminus c_1)$}
                \EndIf
                    \State {$K_0:= K_0 + 1\; mod\; M_0$}
                    \State {$X_0 := copy$}
                    \State {\textbf{exit}}
            \EndIf
        \EndIf
        \If{$K_0 = M_0 - 1$}
            \If{$X_0 = copy$}
                \If{$P_0 = K_0 - 1$}
                    \State {$c_{0}:= (c_{1}\rightarrowtail_L c_0)$}
                \EndIf
                \State {$X_0 := erase$}
            \Else
                \If{$P_0 = K_0$}
                    \State {$c_{0}:= (c_{0}\setminus c_{-1})$}
                \EndIf
            \If{$P_0=0$}
                \State {$c_{0}:= \delta^*_{(0,1)}(c_{0})$}
            \EndIf
            \If {$0<P_0<M_0-1$}
                \State {$c_{0}:= \delta^*_{(1,1)}(c_{0})$}
            \EndIf
            \If {$P_0=M_0-1$}
                \State {$c_{0}:= \delta^*_{(1,0)}(c_{0})$}
            \EndIf
            \State {$K_0:= K_0 + 1\; mod\; M_0$}
            \State {$X_0 := copy$}
            \EndIf
        \EndIf
\EndIf
\EndProcedure
\end{algorithmic}
}
\end{algorithm}

\section{The experiments and the simulation results}
The model presented here is a CAM which is intrinsically parallel. We first implement the model using the programming language Python since it is a high level language making the implementation of the model faster and easier. Indeed, during the phase of the development of writing the code, we use an object-oriented philosophy, especially while passing from the single-lane to multi-lane case. Thus, it is easier to modify and to rewrite some parts of the code to tune it better. Using this code we run a series of experiments for a first evaluation of the model. Part of the data obtained by these experiments have been presented in the EWGT 2012 conference \cite{yeldan-colorni-lue-rodaro}, and the complete version is contained in the PhD thesis \cite{yeldan}. We also implement the model using CUDA to take advantage of the power of the modern graphic cards to perform parallel computation. Indeed, the single-lane model can be naturally parallelize simply by giving a cell (a vehicle in our case) to each thread of the GPU. In Section \ref{sec: conclusion} we give a brief description of this process.


\subsection{Setting the kinds of vehicles and the experiment scenarios}
In this section, we describe the setting up conditions for the kinds of vehicles used in a series of simulations. In our experiments, we use just two kinds of vehicles which we call as \textit{passenger vehicles} and \textit{long vehicles}. The vehicles have the following parameters:
\begin{itemize}
  \item \textit{passenger vehicles}: $v_{max}=36$ m/s, $v_{opt}=28$ m/s, the length $4$ m, the natural acceleration noise: $0.2$ m/s$^2$ (see \cite{greenwood}), $s_{max}=500$ m, $s_{min}=-450$ m, the function of the probability of lane-changing to the right lane $P_R(x)=x$, the function of the probability of lane-changing to the left lane $P_L(x)=x$.
   \item \textit{long vehicles}: $v_{max}=25$ m/s, $v_{opt}=20$ m/s, the length $9$ m, the natural acceleration noise: $0.1$ m/s$^2$ (see \cite{greenwood}), $s_{max}=300$ m, $s_{min}=-700$ m, the function of the probability of lane-changing to the right lane $P_R(x)=x$, the function of the probability of lane-changing to the left lane $P_L(x)=x^{1.25}$.
\end{itemize}
\begin{figure}[H]
 \begin{center}
 \includegraphics[scale=0.7]{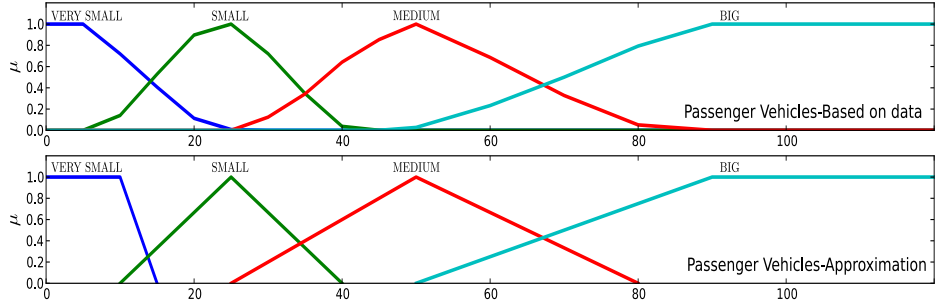}
 \end{center}
  \caption{The membership functions of the passenger vehicles for the front distance}
  \label{MemfuncFD}
\end{figure}
Regarding the fuzzy membership functions, we tune them according to a questionary posed to a group of drivers. The results represent how much the claims are related to each interviewee's perception while driving. We then consider the corresponding membership functions by interpolating linearly these data. However, in the simulations, we do not use these functions based on the data. Instead, we approximate them with the triangular and trapezoidal membership functions which are more common and easier from the computational point of view. We notice from the experiments that this kind of approximation do not alter the results of the simulation significantly. Firstly, because the sample is too small so the data obtained is subject to errors, and secondly, the obtained membership functions are not so different from their approximated versions and the simulation can gain a boost of a factor $1.5$. In Figure \ref{MemfuncFD}, we see one of the membership functions of the passenger vehicles for the perception of front distance adjusted according to the collected data and the approximated version, respectively.

For the first evaluation of our model, we simulate a piece of highway near an off-toll plaza. The parameters used in the simulations are:
\begin{itemize}
  \item Length of the road $L$
  \item Number of lane $M$
  \item Number of iterations (the simulation time)
  \item Number of repetitions of the same experiment
  \item Emission rate $\lambda$: the number of vehicles entering this piece of highway
  \item Influence radius $\rho$ of the off-toll plaza, note that $\rho=-1$ means that there is an open road tolling
  \item Obstacle: An obstacle can be placed on the left-most lane or on the right-most lane
\end{itemize}
We model the entrance of the each lane using a Poisson stochastic process. The probability of emitting at least one vehicle in a lane in the time interval [0,t] is calculated by a Poisson distribution. If the random test succeeds, i.e., if the result is ``emit a vehicle" then we randomly choose a kind of vehicle where the choice is a ``passenger vehicle" with the probability of $1-p$ and a ``long vehicle" with the probability of $p$, and in the experiments we consider $p = 0\%, 10\%, 20\%, 30\%$. The emission rate is set to $\lambda/M$ for each lane where we consider $\lambda = 0.25, 0.5, 1, 1.5, 2$ veh/s. For instance, if $\lambda = 1.5$ and $M=3$, the average number of vehicles entering the piece of highway is $90$ veh/min which means that each lane at the beginning of this road is charged of around $30$ veh/min corresponding to a situation of heavy traffic since the maximum capacity of a lane is considered around $40$ veh/min \cite{kutz}. Regarding the off-toll plaza, we introduce a parameter called influence radius which is the interval of the road segment in which the vehicles are captured when they are exiting. This parameter is used to control the processing time of each vehicle which influences the throughput (the number of vehicles processed in every $10$ seconds). The influence radii considered are $\rho = 10,25,50$ m and ``-1" which means that the off-toll plaza is not visible to the vehicles, thus they simply act as if there is no any off-toll plaza. This situation corresponds to the case of an open road tolling payment system where the vehicles do not need to slow down and stop to make payment. Every $10$ seconds the simulator checks the number of vehicles processed (exited) and it calculates the average latency, i.e., the average of the time spent to travel from the entrance to the exit. We also place an obstacle in some experiments either on the left- or on the right-most lane with a dimension of $2L/5$ in the middle of the piece of highway to analyze a bottleneck phenomenon. In the simulations, each experiment has $1000$ iterations and it is repeated $100$ times. The execution time takes around $6.6$ hours using one core of a computer equipped with a $16$ core Xeon at $2.7$ GHz ($X5550$) with $16$ GB of RAM running Debian Linux (kernel 2.6.32).

\subsection{Analysis of the experimental results}
The macroscopic parameters considered for the purpose of analyzing the general behavior of traffic are evaluated as following. The density describing the number of vehicles per unit length of the piece of highway (measured in vehicles per meter) at time $t$ is calculated by $D(t)=\frac{N(t)}{L}$ where $N(t)$ is the total number of vehicles at time $t$ and $L$ is the length of the road representing a piece of highway ($5$ km long) near the off-toll plaza. The number of vehicles is not constant since the entering rate of the vehicles (depending on a Poisson stochastic process) is different from the exiting one, so we obtain different densities. In our experiments, the density can reach at $0.25$ veh/m/lane as maximum since the length of a passenger vehicle is assumed to be $4$ m long. The average velocity, i.e., the averaged sum of the velocities at time $t$ is calculated by $v_{av}(t)=\sum_{i=1}^{N(t)}v_i/N(t)$, and finally, the flow is evaluated by $q(t)=D(t)\cdot v_{av}(t)$.

The analysis of traffic flow is typically performed by constructing the fundamental diagram (the flow-density diagram) that determines the traffic state of a roadway by showing the relation between flow and density. Therefore, the computations described above are used to plot the fundamental diagrams where we analyze different traffic phenomena, see for instance Figure \ref{flowDiag}.
\begin{figure}
 \begin{center}
 \includegraphics[scale=0.8]{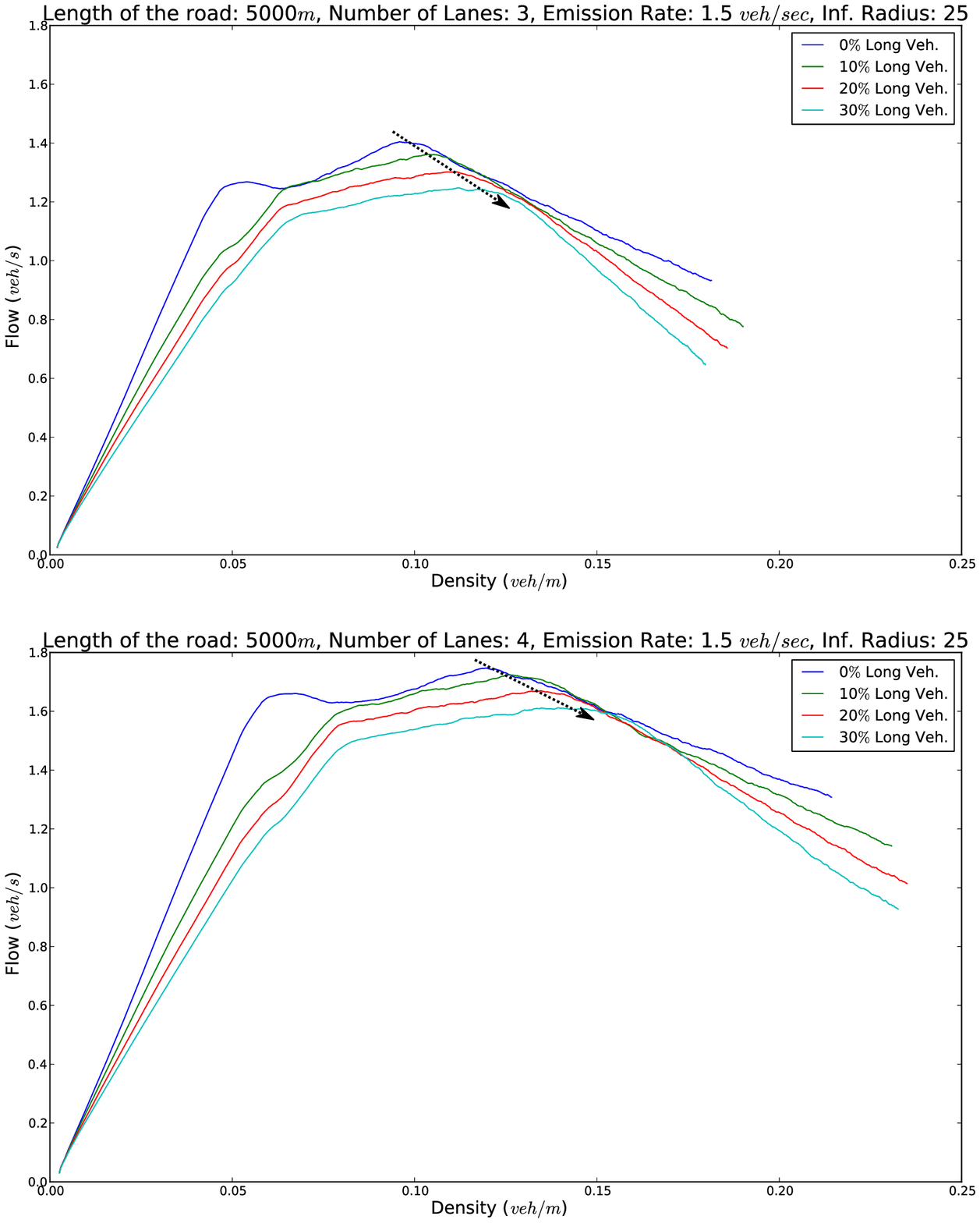}
 \end{center}
  \caption{Comparison of fundamental diagrams with 3 and 4 lanes, respectively}
  \label{flowDiag}
\end{figure}

In the simulation results, we notice that the different composition of vehicles obtained by changing the percentages of the long vehicles is an important factor in influencing the flow. It is seen in Figure \ref{flowDiag} that adding even a small percentage of long vehicles to the traffic stream changes the fundamental diagram significantly. The heterogeneity effects also the throughput as it is predictable since long vehicles are slower and so in the queue near the off-toll plaza it takes more time to move and to get processed.

The experiments show that the model is able to reproduce the typical traffic flow physical phenomena such as the three phases of traffic flow \cite{kerner-rehborn}: Free flow, synchronized flow and wide-moving jam, see Figure \ref{flow-cross-5000-3-2-10}. Free flow corresponds to the region of low to medium density and weak interaction between vehicles. In general, the slope of the fundamental diagram in the free flow phase is related to the speed limit, meaning that in this phase vehicles can move almost at the speed limit. Instead, in the free flow phase of the fundamental diagram in our experimental results, the slope is related to the optimal velocity, meaning that in this phase the vehicles can move almost at the optimal velocity. The reason is that in our model the vehicles do not aim to reach to the maximum velocity as it is in the NaSch-type models, but they tend to go with their optimal velocity. Free flow is characterized with a strong correlation and quasi-linear relation between the local flow and the local density \cite{neubert-santen-schad-screck}. 
\begin{figure}
 \begin{center}
 \includegraphics[scale=0.75]{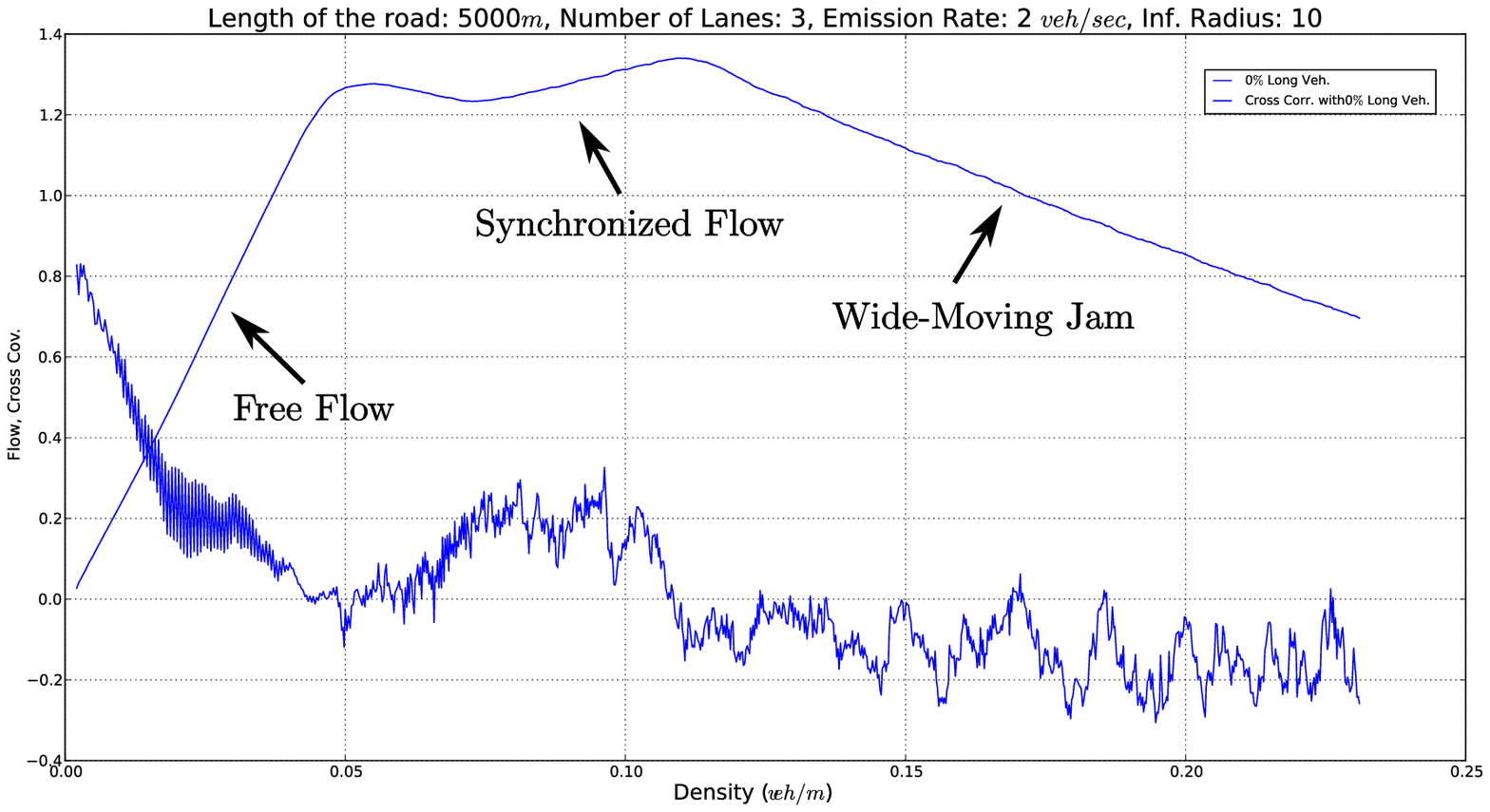}
 \end{center}
  \caption{Traffic phases in the fundamental diagram: Free flow, synchronized flow and wide-moving jam, and the cross-covariance between the flow and density}
  \label{flow-cross-5000-3-2-10}
\end{figure}
The synchronized flow presents medium and high density while the flow can behave free or jammed. In other words, it is defined by the interaction between the vehicles and is characterized by an uncorrelated flow-density diagram. However, this phase is not clearly understood in the context of CA and not observed in most of the NaSch-type models. It probably requires the presence of sources (on-ramps, on-toll plaza) and sinks (off-ramps, off-toll plaza), see \cite{wolf}. The wide-moving jam phase represents the situation where the traffic is jammed (congested). In this phase, an increase in the density results with a decrease in the flow. Let us now consider the cross-covariance between the flow $q(t)$ and the density $D(t)$, see Figure \ref{flow-cross-5000-3-2-10}:
$$
cc(q,D)=\frac{\langle q(t)D(t)\rangle-\langle q(t)\rangle \langle D(t)\rangle}{\sqrt{\langle q(t)^2\rangle-\langle q(t)\rangle^2}\sqrt{\langle D(t)^2\rangle-\langle D(t)\rangle^2}}
$$
where the brackets $\langle \cdot \rangle$ indicate averaging the values obtained in all the experiments at time $t$. In the free flow phase, the flow is strongly related to the density indicating that the average velocity is nearly constant. For large densities, in the wide-moving jam phase, the flow is mainly controlled by density fluctuations. There is a transition between these two phases where the fundamental diagram shows a plateau when the cross-variance is close to zero. This situation with $cc(q,D)\approx 0$ is identified as the synchronized flow \cite{knospe-santen-schad-schrec2,neubert-santen-schad-screck}.

It is seen in Figure \ref{flow-cross-5000-3-2-10} that we reproduce these relations where the free flow phase initiates with the situation of cross-covariance close to $1$ and continues with a positive cross-covariance, the synchronized flow phase is in the region where the cross-covariance is close to zero and the wide-moving jam phase corresponds to the situation where the cross-covariance is negative (anticorrelation).

The plateau formation has the dependency also on the throughput. In Figure \ref{comp-5000-3-1half-0-20-30long}, it is seen this situation where we have compared the fundamental diagrams for different influence radii which are related to the average throughputs. Observing any of these three plots, we see that when the average throughput is increased, there occurs a more immediate transition from the synchronized flow phase to the wide-moving jam phase. In other words, the phase-change between these two flows occurs with a higher density with the decrease of the average throughput. Furthermore, when we compare the three plots in Figure \ref{comp-5000-3-1half-0-20-30long}, we see that with more long vehicle percentage we have the phase-change with a higher density. The same phenomenon is also shown with arrows in Figure \ref{flowDiag}.

\begin{figure}
 \begin{center}
 \includegraphics[scale=1]{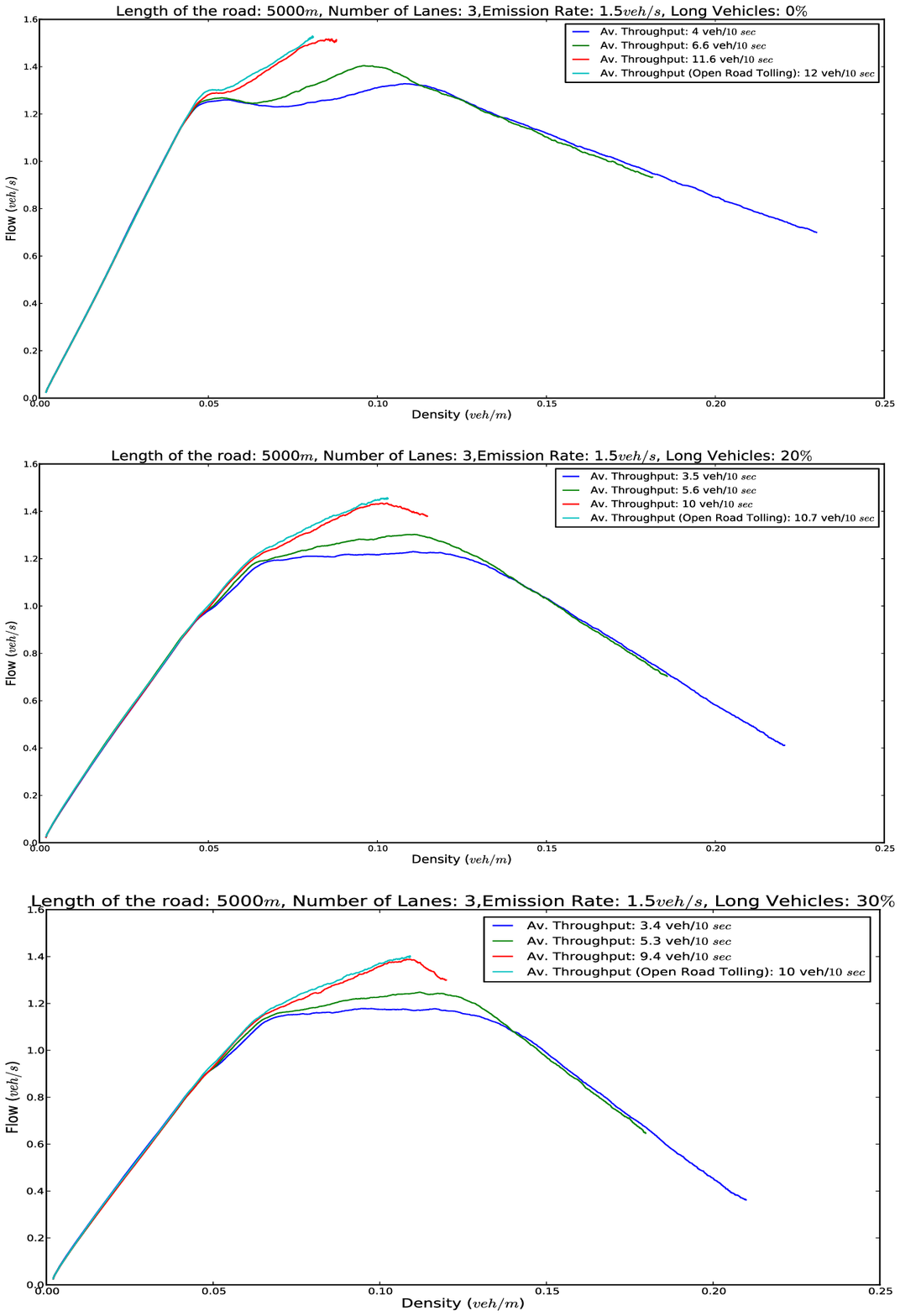}
 \end{center}
  \caption{The fundamental diagrams depending on the different average throughputs}
  \label{comp-5000-3-1half-0-20-30long}
\end{figure}

\begin{figure}
 \begin{center}
 \includegraphics[scale=0.8]{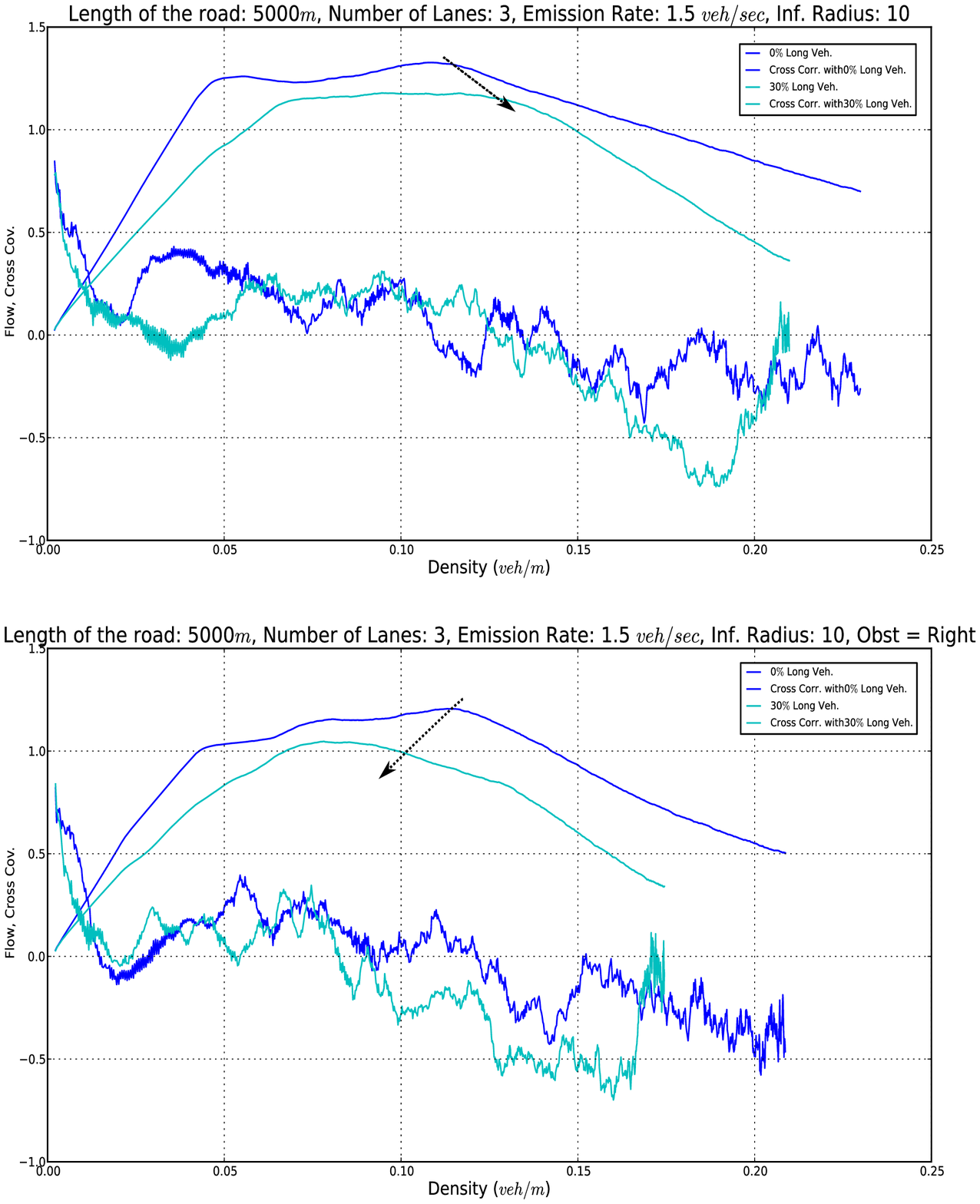}
 \end{center}
  \caption{The fundamental and the cross-covariance diagrams without and with obstacle, respectively}
  \label{obstacle-5000-3-1half-10}
\end{figure}

\begin{figure}
 \begin{center}
 \includegraphics[scale=0.8]{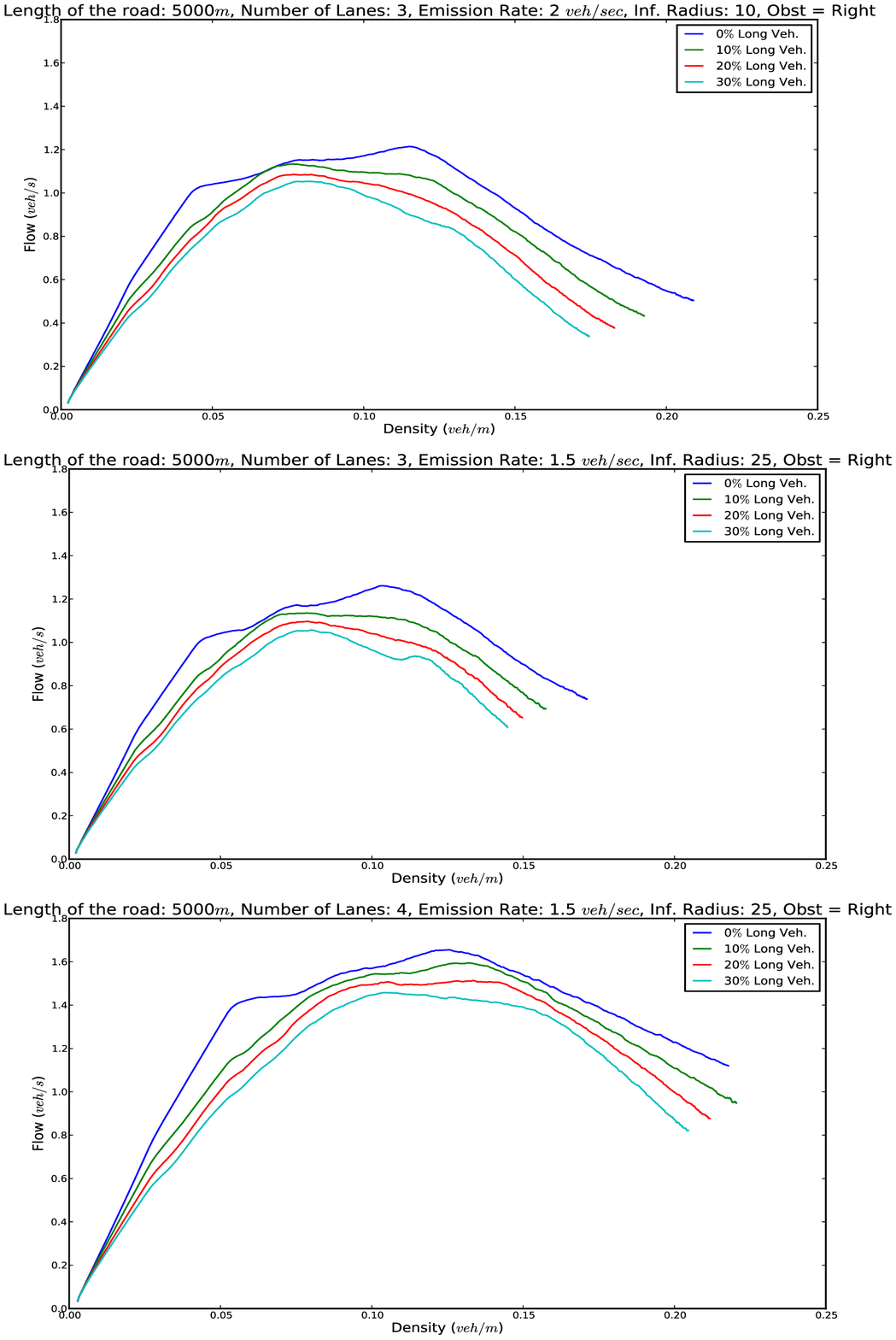}
 \end{center}
  \caption{The slope in the wide-moving jam phase with the obstacles}
  \label{slope-congestion}
\end{figure}

 \begin{figure}
 \begin{center}
 \includegraphics[scale=0.8]{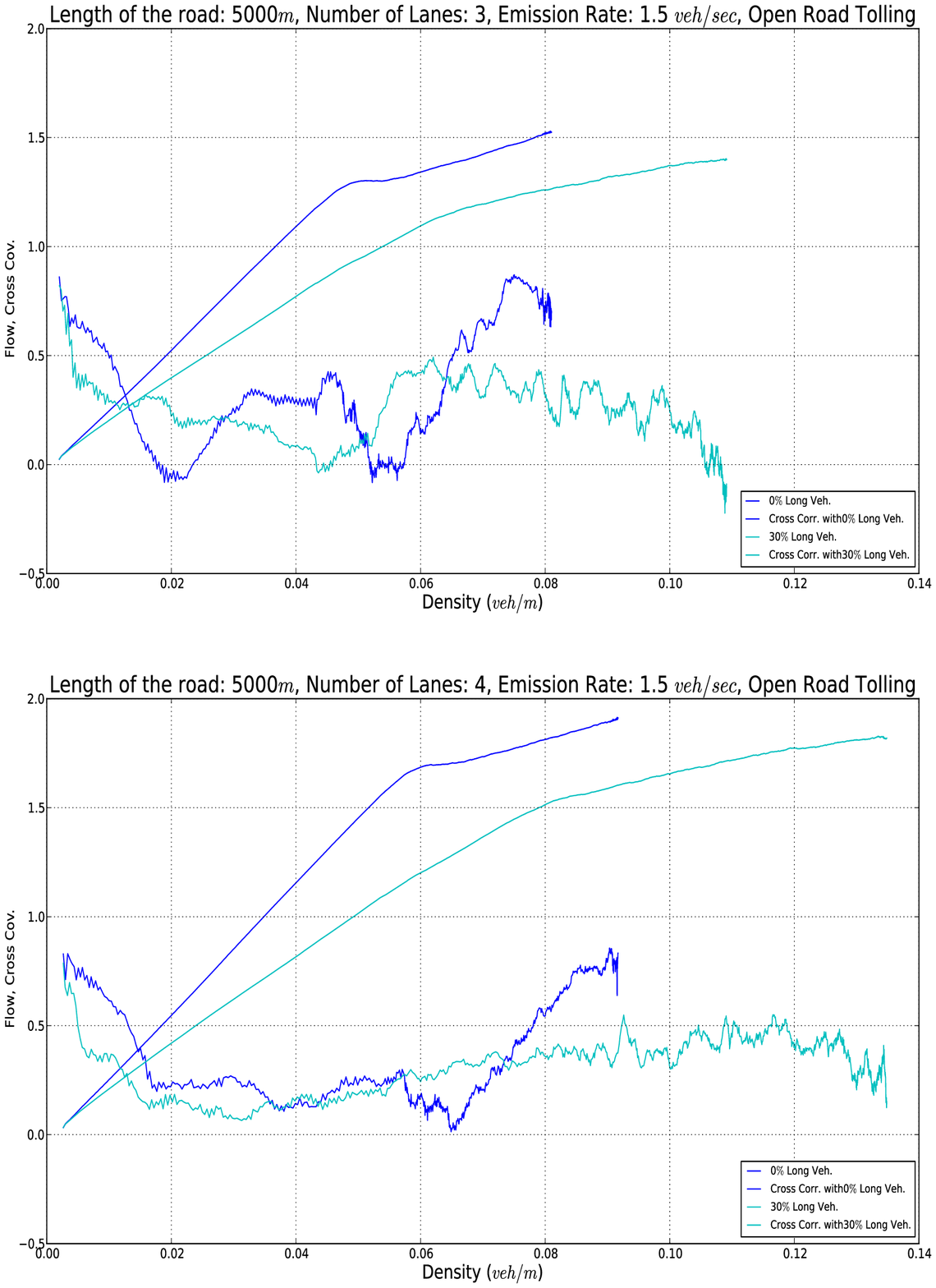}
 \end{center}
  \caption{The effect of open road tolling on the flow phases with 3 lanes and 4 lanes}
  \label{5-3and4-1half-open}
\end{figure}

This phenomenon is inverted in the presence of an obstacle (on the right-most lane\footnote{When we place an obstacle on the left-most lane, we observe that the result is almost the same with placing it on the right-most lane, so when we make comparisons of having or not having obstacle we use the result of placing it on the right-most lane.}). More precisely, when an obstacle is set we observe that the phase-change occurs with a lower density, see Figure \ref{obstacle-5000-3-1half-10}. This phenomenon is more clear with the presence of long vehicles, probably because when long vehicles get stuck in traffic this phase emerges faster. Besides, without an obstacle there is an increase at the slope (absolute value) of the fundamental diagram where there is the wide-moving jam phase, when the percentage of long vehicles is increased (see Figure \ref{flowDiag}). In other words, the flow decreases faster in the presence of long vehicles with the same increment of density, which is not observed in the case of placing an obstacle (the lines are almost parallel, see Figure \ref{obstacle-5000-3-1half-10}). However, this parallelism occurs only when the road is enough saturated. For instance in Figure \ref{slope-congestion} we see that the parallelism occurs in the first and second plots, but not in the third plot since the lane number is increased and so the emission rate per lane is decreased, i.e., there is not enough saturation.

Another effect of setting an obstacle is a reduction on the traffic capacity as it is expected. For instance, in the experiments plotted in Figure \ref{obstacle-5000-3-1half-10}, the maximum flow that can be reached is around $1.33$ veh/sec without an obstacle and around $1.2$ veh/sec with an obstacle placed on the right-most lane. In this figure, we also see that this reduction is more evident with the presence of long vehicles.

In the case that there is an open road tolling, the wide-moving jam phase does not take place, so there is no anticorrelation situation. As it is seen in Figure \ref{5-3and4-1half-open}, there is just a transition between the free flow and the synchronized flow. The absence of wide-moving jam phase occurs in general with the situations where the emission rate is low with respect to the rate that the vehicles are processed (the influence radius or throughput) at the off-toll plaza as it is expected.

The experiments also show that the model is able to reproduce the hysteresis phenomena in transition between free flow and synchronized flow phases (see \cite{kerner-rehborn}) and in transition between synchronized flow and wide-moving jam phases. The regions where there are saddles followed by a capacity drop, as in the Fig.9(e) of \cite{geroliminis-sun}, are the regions of metastability in the phase-changes (from free flow to synchronized flow and from synchronized flow to wide-moving jam). This metastability phenomenon is not so evident when there are long vehicles. More precisely, the heterogeneity of traffic effects the formation of the plateaus. In Figure \ref{flowDiag} and \ref{comp-5000-3-1half-0-20-30long}, we see that in the case of the presence of long vehicles, the bumpy plateaus in the synchronized flow phase are replaced by more flattened plateaus, so the saddles in the phase-changes are not observed as they are in the absence of long vehicles. This is probably due to the fact that passenger vehicles are faster and so the flow of a traffic without long vehicles changes its phase more sharply. However, we observe that the presence of obstacles increases the number of the saddles even when there are long vehicles, making the fundamental diagrams more bumpy especially in the synchronized flow phase (see Figure \ref{obstacle-5000-3-1half-10} and \ref{slope-congestion}).

\section{Conclusion}\label{sec: conclusion}

It is introduced a new model for multi-lane traffic. This model is an attempt to define a new class of traffic models which is a hybrid between the usual microscopic traffic models, like car-following models, and the usual CAMs which we have identified as the NaSch-type models. In the process of the extension of the single-lane model to the multi-lane case, we have first presented the model as an array of communicating one-dimensional CCA, and then we have proved that this model can be simulated by a suitable CCA. In this way, we have framed our multi-lane model inside the class of CCA.

For a first test, we have implemented the model using Python\footnote{The code can be found in the public Dropbox link \myurl.} with an object-oriented philosophy of programming. Using a questionnaire we have set up two kinds of vehicles which we have used to run a series of experiments. Analyzing the experimental results, we have studied the influence of different composition of vehicles on the macroscopic behavior of the traffic in order to observe the typical traffic flow phenomena. We have found that the results are promising, indeed we reproduce and enrich the fundamental diagrams of traffic flow.

\begin{table}[h]
\caption{Computation Time Comparison between CPU and GPU}
\begin{center}
\scalebox{0.75}{
    \begin{tabular}{ | c | c | c | }
    \hline
    Number of Vehicles & CPU & GPU \\
     (per lane)  & (sec) & (sec)  \\ \hline
       0;    0;  1500  &  484 &   45,8 \\ \hline
       0;    0;  3000  &  958 &  104   \\ \hline
       0;    0;  5000  & 1608 &  194   \\ \hline
       0;    0; 10000  & 4270 &  556   \\ \hline
    5000; 5000;  5000  & 9679 & 1288   \\ \hline
    \end{tabular}
}
\end{center}
\label{tab:comparisonCUDAPython}
\end{table}

The code written in Python does not take advantage of CA and its typical synchronous behavior. For this reason, we have adapted the code using PyCuda\footnote{The code can be found in the public Dropbox link \myurlcuda.} to partially parallelize Algorithm \ref{alg:evolve multi} on GPU's and we have seen that it is possible to boost the speed of execution by a factor of $\sim 10$. Indeed, in Table \ref{tab:comparisonCUDAPython} we see the computation time comparison between the two codes made by running $1000$ steps of simulation for a road with $3$ lanes and where it is given the initial numbers of the vehicles distributed per lane (on a laptop equipped with a processor $i7$ intel and with a graphic card NVIDIA GeForce GT $555$M). The increased speed of execution by higher factors of simulation is important also to adapt the model for a forecasting usage. By partial parallelization, we mean that the lane-changing process performed using the operations $\rightarrowtail_L, \rightarrowtail_R, \setminus$ is done sequentially, and the only parts that can be run in parallel are the global transition functions $\delta_{(a,b)}^*,\;a,b\in \{0,1\}$. Therefore Algorithm \ref{alg:evolve multi} can be completely parallelized only if we decide not to apply the precedence rules and assume a concurrent strategy of lane-changing.

The analysis we perform here is not conclusive, but gives an insight of the potentiality of our model. For this reason, we suggest the following tasks as future works and research directions to improve and validate the model and the simulator:

\begin{itemize}
  \item We did not use all the potentiality of the code since our aim was to give a first evaluation to our model. However, it would be interesting to consider also the experiments involving on- and off-ramps and loop-detectors to analyze different and more realistic situations.
  \item The heterogeneity considered is reduced to two kinds of vehicles. A natural question is how the system reacts introducing other kinds. For instance, in highway environments, motorcycles or sport vehicles can be added to the mixed traffic of passenger vehicles and long vehicles. This also brings with it the interesting issue of how to tune the membership functions for such new kinds.
  \item The process of lane-changing is purely stochastic. However, in literature there are some attempts in microscopic models where the process of lane-changing is described by using a fuzzy logic-based system \cite{brackstone-mcdonald,errampalli-okushima-akiyama}, thus it would be interesting to extend our model to a model in which it is implemented a fuzzy logic-based system to refine the lane-changing rules.
  \item The model has to be compared with real data. In other words, a careful case study on specific scenarios with the data available is necessary for the validation by the community of people working on traffic flow theory, granular flow theory and traffic (transportation) engineering.
\end{itemize}

\section*{Acknowledgments}
The author acknowledges the support from the European Regional Development Fund through the programme COMPETE and by the Portuguese Government through the FCT -- Funda\c c\~ao para a Ci\^encia e a Tecnologia under the project PEst-C/MAT/UI0144/2011. The second author also acknowledges the support of the FCT project SFRH/BPD/65428/2009.

\bibliographystyle{plain}
\bibliography{citations}

\end{document}